\documentclass[a4paper,UKenglish,cleveref, autoref, thm-restate]{lipics-v2021}
\pdfoutput=1 %uncomment to ensure pdflatex processing (mandatatory e.g. to submit to arXiv)
\hideLIPIcs  %uncomment to remove references to LIPIcs series (logo, DOI, ...), e.g. when preparing a pre-final version to be uploaded to arXiv or another public repository

\title{Improved Hardness Results for the Guided Local Hamiltonian Problem}
\titlerunning{} %TODO optional, please use if title is longer than one line

\author{Chris {Cade}}{QuSoft \& University of Amsterdam (UvA), Netherlands}{chris.cade@cwi.nl}{}{}
\author{Marten {Folkertsma}}{QuSoft \& CWI, Netherlands}{marten.folkertsma@cwi.nl}{}{}
\author{Sevag {Gharibian}}{Paderborn University, Germany}{sevag.gharibian@upb.de}{}{}
\author{Ryu {Hayakawa}}{Kyoto University, Japan}{ryu.hayakawa@yukawa.kyoto-u.ac.jp}{}{}
\author{Fran{\c c}ois {Le Gall}}{Nagoya University, Japan}{legall@math.nagoya-u.ac.jp}{}{}
\author{Tomoyuki Morimae}{Kyoto University, Japan}{tomoyuki.morimae@yukawa.kyoto-u.ac.jp}{}{}
\author{Jordi {Weggemans}}{QuSoft \& CWI, Netherlands}{jordi.weggemans@cwi.nl}{}{}

\authorrunning{C. Cade, M. Folkertsma, S. Gharibian, R. Hayakawa, F. Le Gall, T. Morimae and J. Weggemans}
\Copyright{Chris Cade, Marten Folkertsma, Sevag Gharibian, Ryu Hayakawa, Francois Le Gall, Tomoyuki Morimae and Jordi Weggemans}
 %TODO mandatory: Please choose ACM 2012 classifications from https://dl.acm.org/ccs/ccs_flat.cfm \ccsdesc{Theory of computation~Quantum complexity theory}

 %TODO mandatory; please add comma-separated list of keywords

\category{Track A: Algorithms, Complexity and Games} %optional, e.g. invited paper

\relatedversion {This paper is based on our following
two technical reports}
\relatedversiondetails{Previous Version}{https://arxiv.org/abs/2207.10097}
\relatedversiondetails{Previous Version}{https://arxiv.org/abs/2207.10250}

\acknowledgements{
We thank Jonas Helsen for feedback on an earlier draft, and Ronald de Wolf for helpful comments. CC acknowledges support from QuSoft and CWI, as well as the University of Amsterdam (UvA) under a POC (proof of concept) fund. MF and JW were supported by the Dutch Ministry of Economic Affairs and Climate Policy (EZK), as part of the Quantum Delta NL programme.
RH, FLG, and TM were supported by  MEXT Quantum Leap Flagship Program (MEXT Q-LEAP) grant No.~JPMXS0120319794. RH was also supported by JSPS KAKENHI Grant Number JP22J11727. FLG was also supported by JSPS KAKENHI grants Nos.~JP19H04066, JP20H05966, JP20H00579, JP20H04139, JP21H04879. SG was supported by DFG grants 450041824 and 432788384.
TM was supported by JST Moonshot JPMJMS2061-5-1-1, JST FOREST, the Grant-in-Aid for Scientific Research
(B) No.JP19H04066, the Grant-in Aid for Transformative Research Areas (A)
21H05183, and the Grant-in-Aid for Scientific Research (A) No.22H00522.}%optional

\nolinenumbers %uncomment to disable line numbering

\EventEditors{Kousha Etessami, Uriel Feige, and Gabriele Puppis}
\EventNoEds{3}
\EventLongTitle{50th International Colloquium on Automata, Languages, and Programming (ICALP 2023)}
\EventShortTitle{ICALP 2023}
\EventAcronym{ICALP}
\EventYear{2023}
\EventDate{July 10--14, 2023}
\EventLocation{Paderborn, Germany}
\EventLogo{}
\SeriesVolume{261}
\ArticleNo{93}
\usepackage{amsmath,amsfonts,amssymb,tikz,amsthm}
\usepackage{graphicx}
\usepackage{color}
\usepackage{mathtools}
\usepackage{multirow}
\usepackage{array}
\usepackage{comment}
\usepackage{bbm}
\usepackage{ascmac}
\usepackage{fancybox}
\usepackage[utf8]{inputenc}
\usepackage{braket}

\newcommand{\GLHdes}{\textup{GLH}}
\newcommand{\BQP}{\textup{BQP}}

\newcommand{\QMA}{\textup{QMA}}
\newcommand{\QCMA}{\textup{QCMA}}
\newcommand{\complex}{{\mathbb C}}

\newcommand{\glh}[0]{\text{GLHLE}(k,c,a,b,\delta)}
\newcommand{\mO}{\mathcal{O}}

\newcommand{\psihist}{\psi_{\textup{hist}}}
\newcommand{\hin}{H_{\textup{in}}}
\newcommand{\hprop}{H_{\textup{prop}}}
\newcommand{\hstab}{H_{\textup{stab}}}
\newcommand{\hout}{H_{\textup{out}}}

\newcommand{\abs}[1]{\left|#1\right|}
\newcommand{\norm}[1]{\|#1\|}
\newcommand{\poly}{\mathrm{poly}}
\newcommand{\tr}{\mathrm{tr}}

\makeatletter
\newcommand{\newreptheorem}[2]{\newtheorem*{rep@#1}{\rep@title}\newenvironment{rep#1}[1]{\def\rep@title{#2 \ref*{##1}}\begin{rep@#1}}{\end{rep@#1}}}
\makeatother
\newcommand{\kb}[2]{|#1\rangle\langle#2|}

\begin{document}

\maketitle

\begin{abstract}
Estimating the ground state energy of a local Hamiltonian is a central problem in quantum chemistry. In order to further investigate its complexity and the potential of quantum algorithms for quantum chemistry, Gharibian and Le Gall (STOC 2022) recently introduced the \emph{guided local Hamiltonian problem (GLH)}, which is a variant of the local Hamiltonian problem where an approximation of a ground state (which
    is called a guiding state) is given as an additional input. Gharibian and Le Gall showed quantum advantage (more precisely, \BQP-completeness) for GLH with $6$-local Hamiltonians when the guiding state has fidelity (inverse-polynomially) close to $1/2$ with a ground state.

    In this paper, we optimally improve both the locality and the fidelity parameter: we show that the \BQP-completeness persists even with 2-local Hamiltonians, and even when the guiding state has fidelity (inverse-polynomially) close to 1 with a ground state. Moreover, we show that the \BQP-completeness also holds for 2-local physically motivated Hamiltonians on a 2D square lattice or a 2D triangular lattice. 
    Beyond the hardness of estimating the ground state energy, we also show \BQP-hardness persists when considering estimating energies of \emph{excited} states of these Hamiltonians instead.
    Those make further steps towards establishing \emph{practical quantum advantage} in quantum chemistry.
\end{abstract}

\section{Introduction}
Simulation of physical systems is one of the originally envisioned applications of quantum computing~\cite{F82,F85}. Quantum chemistry, in particular, has seen much activity on this front in recent years, e.g. \cite{Aaronson09, Aspuru-Guzik+05, Bauer+20, Lee+21, Reiher+17,YBWRB21}. There, a central goal is to estimate the {ground state energy} of a given {$k$-local Hamiltonian} $H$, denoted the \emph{$k$-local Hamiltonian problem ($k$-LH)}. Roughly, for this problem, a $k$-local Hamiltonian $H=\sum_i H_i$ on $n$ qubits is a $2^n\times 2^n$ Hermitian matrix, specified succinctly via ``local quantum clauses'' $H_i$ acting on $k\in O(1)$ qubits each. The eigenvalues of $H$ are the discrete energy levels of the corresponding quantum system. In particular, the smallest eigenvalue, which we denote $\lambda_0(H)$, is called the ground state energy. An eigenvector corresponding to $\lambda_0(H)$ is called a ground state, and describes a state of the quantum system in the energy configuration $\lambda_0(H)$. Note that $k$-LH strictly generalizes classical $k$-SAT, in that any instance of the latter can be embedded into the former.

Unfortunately, it is nowadays well-known that estimating ground state energies of local Hamiltonians is QMA-complete~\cite{Kitaev+02}. This hardness persists, moreover, even in the bosonic~\cite{WMN10} and fermionic settings~\cite{SV09}. Thus, assuming $\BQP\neq \QMA$, one cannot hope for an efficient algorithm for $k$-LH on \emph{all} $k$-local Hamiltonians.

\paragraph*{What actually happens in practice.} In an attempt to bypass worst-case hardness results, in practice the quantum chemistry community often adopts the following two-step procedure:
\begin{itemize}
\item (Step 1: Ground state approximation) A classical heuristic algorithm is applied to obtain a ``guiding state'' $\ket{\psi}$, which is hoped to have ``good'' fidelity with a ground state.
 \item (Step 2: Ground state energy approximation) The guiding state $\ket{\psi}$ is used in Quantum Phase Estimation (QPE) \cite{Kitaev95} to efficiently compute the corresponding ground state energy \cite{Abrams+99,Aspuru-Guzik+05}. (A more recent approach is based on variational quantum algorithms, aimed more at near-term hardware (see \cite{Cerezo+21} for a survey), but which is heuristic in nature (unlike QPE).)
\end{itemize}

\noindent Two comments: (1) There is something special about Step 2 --- it is a unique strength of quantum computers to be able to resolve an eigenvalue (within additive $1/\poly(n)$ precision) of a (sparse) Hermitian matrix given just an \emph{approximation} $\ket{\psi}$ to the corresponding eigenvector (via QPE)!\footnote{
Actually, quantum computers can efficiently prepare a ground state with fidelity $1-1/\exp(n)$ given access to a guiding state $\ket{u}$ that has inverse polynomial fidelity with a ground state $\ket{g}$ (i.e. $|\langle u|g\rangle|\geq 1/\poly(n)$)
using quantum amplitude amplification for local Hamiltonians that have inverse-polynomial spectral gaps~\cite{lin2020near}.
}
Indeed, the closely related task of (sparse) matrix inversion, which can be solved efficiently on a quantum computer coherently by diagonalizing the matrix and ``manually'' inverting its eigenvalues via postselection, is \BQP-complete~\cite{HHL09}. % (which can be viewed as coherently diagonalizing the matrix and ``manually'' inverting its eigenvalues via postselection.)
(2) In general, one does \emph{not} expect a good\footnote{``Good'' here meaning a state $\ket{\psi}$ with inverse polynomial fidelity with a ground state, and with a succinct classical description allowing $\ket{\psi}$ to be prepared efficiently.} guiding state for arbitrary local Hamiltonian $H$ to exist, as this would imply $\QCMA=\QMA$. And even if such a guiding state \emph{did} exist, finding it can still be hard. For example, minimizing $\tr(H\rho)$ over the ``simplest'' quantum ansatz of tensor product states, i.e. $\rho=\rho_1\otimes\rho_2\otimes\cdots\otimes\rho_n$ for $\rho_i\in\mathcal{L}(\complex^2)$, remains NP-hard (seen by letting $H$ be a diagonal Hamiltonian encoding a classical $3$-SAT instance).
%{\color{red}Tomoyuki:Is it NP-hard even for $1/poly$ approximation?}

\paragraph*{Directions for study.} With Steps 1 and 2 above in mind, in order to practically obtain a quantum advantage for quantum chemistry problems, there are two branches of study necessary:
\begin{itemize}
     \item (Step 1: Ground state approximation) Here, the best one can hope for is fast algorithms tailored to physically motivated \emph{special cases} of Hamiltonians $H$ (either heuristic or worst-case poly-time complexity). This is arguably the bottleneck for fast quantum algorithms outperforming classical techniques~\cite{C22}.

     \item (Step 2: Ground state energy approximation) A thorough complexity theoretic understanding of which Hamiltonian families provably permit quantum computers to outperform classical ones, assuming a good guiding state has been found (in Step 1).
\end{itemize}
In \cite{gharibian2021dequantizing}, the formal study of the second step above was initiated. Specifically, the \emph{Guided $k$-local Hamiltonian problem ($k$-GLH)} was introduced, which is stated roughly as follows (formally given in \Cref{def:GLHLE}): {Given a $k$-local Hamiltonian $H$, an appropriate ``representation'' of a guiding state $\ket{\psi}$ with $\delta$-fidelity with the ground space of $H$, and real thresholds $\beta> \alpha$, estimate the ground state energy of $H$.} Then, two results were shown:
\begin{itemize}
\item
%When $\ket{\psi}$ is represented as a ``samplable'' state~\cite{T19},
For any constant $k$, $k$-GLH can be efficiently solved \emph{classically} within \emph{constant} precision, i.e. for $\beta-\alpha\in\Theta(1)$ and $\delta\in \Theta(1)$.
\item
In contrast,
%when $\ket{\psi}$ is represented as a so-called ``semi-classical state'',
$6$-GLH is \BQP-hard for \emph{inverse polynomial} precision, i.e. $\beta-\alpha\geq 1/\poly(n)$, and  $\delta=1/\sqrt{2}-1/\poly(n)$.
\end{itemize}
The latter regime of inverse-polynomial precision turns out to be the relevant one for solving quantum chemistry problems in practice --- the desired ``chemical accuracy'' is about 1.6 millihartree (which is constant relative to an unnormalized Hamiltonian), which upon renormalization of the Hamiltonian (as done here) yields the claimed inverse polynomial precision. This \BQP-hardness result thus gives theoretical evidence for the superiority of quantum algorithms for chemistry.

Four important problems were left open in \cite{gharibian2021dequantizing}: 
Is $k$-GLH still \BQP-hard with larger $\delta$, and in particular for $\delta$ arbitrarily close to 1? 
Is $k$-GLH still \BQP-hard for $k<6$? 
Is $k$-GLH still \BQP-hard for estimating the excited state energies?
Is $k$-GLH still \BQP-hard for physically motivated Hamiltonians?

\paragraph*{This work.} In this work, we continue the agenda toward Step 2 above by resolving these four open questions. Here are our main contributions:
\begin{itemize}
\item
First, we show that \BQP-hardness continues to hold even for $\delta= 1-1/\poly(n)$, i.e. even when we are promised the guiding state $\ket{\psi}$ is a remarkably good approximation to the ground state.
\item
Second, we show that \BQP-hardness continues to hold even for $k=2$. (Note that for $k=1$, the problem can be solved efficiently classically, even without a guiding state.)
\item Third, we extend the $\BQP$-hardness results to the case when one is interested in estimating energies of excited states, rather than just the groundstate. Interestingly, we are only able to show $\BQP$-completeness in this setting by showing that the first point holds, i.e. the \BQP-hardness in the regime $\delta \in [\frac{1}{2}+\Omega(1/poly(n)),1-{\Omega}(1/poly(n))]$.
\item
Fourth, we prove hardness results for \emph{physically motivated} Hamiltonians.
They include
XY model (constraints of the form
$XX+YY$), Heisenberg model (constraints of the form $XX+YY+ZZ$), the antiferromagnetic $XY$ model and the antiferromagnetic Heisenberg model (i.e. ``Quantum Max Cut''~\cite{GP19}).
In contrast, the \BQP-hardness construction of \cite{gharibian2021dequantizing} is arguably artificial, because
they used the circuit-to-Hamiltonian construction of \cite{Kitaev+02} and query Hamiltonian construction of~\cite{A14}.
\end{itemize}

To formalize the third direction, we introduce the \textit{Guided $k$-Local Hamiltonian Low Energy}-problem ($k$-$\text{GLHLE}$) in which the guiding state has $\delta$-fidelity with the $c$'th excited state of $H$ and the problem is to estimate the $c$'th excited state energy of $H$ (for a formal definition, see Definition~\ref{def:GLHLE}). 
Then, the four contributions above are summarized in the following theorem.

\begin{theorem}[Main result]
\label{thm:main}
For any $\delta\in (0,1-\Omega(1/\poly(n)))$, constant $k \geq 2$ and some integer $0\leq c \leq \mathcal{O}(poly(n))$, there exist $a,b\in [-1,1]$ with $b-a\in \Omega(1/\poly(n))$ such that k-$\text{GLHLE}$ is $\BQP$-hard.
Moreover, it is still $\BQP$-hard if the 2-local Hamiltonian is restricted to any of the following families of Hamiltonians:
\begin{itemize}
    \item non-2SLD Hamiltonian on a 2D square lattice
    \item antiferromagnetic Heisenberg model
    \item antiferromagnetic $XY$ model on a 2D triangular lattice.
\end{itemize}
\end{theorem}

Here, the ``non-2SLD'' Hamiltonians are, roughly,
$2$-local Hamiltonians that cannot be diagonalized via single-qubit unitaries (see \Cref{def:non2sld} for the formal definition).
(The term 2SLD is short for ``the 2-local parts of all interactions in the set are simultaneously locally diagonalizable''.)
It was originally introduced in the Hamiltonian complexity classification of Cubitt and Montanaro~\cite{cubitt2016complexity}.
The $XY$ model and the Heisenberg model are examples of non-2SLD Hamiltonians.

\paragraph*{Techniques.}
Now let us explain our technical contributions.
Our first result is the improvement of the fidelity $\delta$
 (Proposition \ref{proposition:overlap} in Section~\ref{sec:complexity}).
The construction of~\cite{gharibian2021dequantizing} cannot exceed $\delta=1/2$ %\jordi{It can get inverse polynomially close to $1/2$ when considering the fidelity instead of the overlap right?}
, but we achieve the fidelity $\delta=1-1/\poly(n)$.
Let us explain why the construction of~\cite{gharibian2021dequantizing} cannot exceed the fidelity $\delta=1/2$.
Their construction for the \BQP-hardness result is the following local Hamiltonian
\[
H= \frac{\alpha+\beta}{2} I\otimes \ket{0}\bra{0} + H' \otimes \ket{1}\bra{1},
\]
where $\beta-\alpha>1/\poly(n)$
and $H'$ is a certain local Hamiltonian whose lowest eigenvalue is $\leq \alpha$ in the YES case and is $\geq \beta$ in the NO case.
It is clear that a ground state of $H$ is
$\ket{\psi}\otimes \ket{1}$ in the YES case, where
$\ket{\psi}$ is a ground state of $H'$.
For the NO case, a ground state is $|0...0\rangle\otimes \ket{0}$.
It can then be easily observed that the optimal guiding state (i.e. the guiding state that has the maximum fidelity with ground states
{\it in both the YES and the NO cases}) is written as $|\phi\rangle\otimes|+\rangle$ for a certain choice of $|\phi\rangle$,
which shows that the fidelity cannot exceed $1/2$ in this construction.

To overcome the problem, we use the perturbation theory approaches of \cite{KKR06,bravyi2011schrieffer}.
In particular, we use first-order perturbation theory, either using the general Schrieffer-Wolf transform framework of \cite{bravyi2011schrieffer} 
%(Section \ref{sec:complexity}) 
or a more first-principles approach via the Projection Lemma.
%(Appendix \ref{appendix:anotherproof})
The main idea is to use a large energy penalty term to rule out all low-energy states which do not look like ``history states''.
We then show that the corresponding guiding state can be chosen as the semi-classical subset state introduced in \cite{gharibian2021dequantizing} (see Definition~\ref{def:subsetstate} in Section~\ref{sec:preliminaries}).
To obtain this, we notice that the ground state of our Hamiltonian is gapped and unique. This is because we are doing a reduction from \BQP\ (as opposed to \QMA). In other words,
there is no \QMA\ ``proof'' to be plugged into the history state construction, and therefore there is a unique low-energy history state. 
In sum, via perturbation theory, we are able to \emph{directly} approximate the ground state with a guiding state in \emph{both} YES and NO cases, as opposed to the block encoding approach of~\cite{gharibian2021dequantizing}, which used equally weighted orthogonal subspaces to \emph{separately} encode the YES and NO cases, respectively.

Our second result is \BQP-hardness of $k$-GLH for $k=2$ (Propositions \ref{prop:locality} in Section~\ref{sec:complexity}).
Here, the universal simulation setup of \cite{cubitt2018universal,zhou2021strongly} cannot be directly applied, because
although their results can approximately preserve the ground space of the input Hamiltonian, it was not known whether
semi-classical subset states can be mapped to semi-classical subset states under such simulation frameworks, and the latter is essential for guiding states used in GLH.
We show that this is indeed the case.
In particular, we show that the original semi-classical subset state of the input 5-GLH instance is mapped to a state with polynomially many ancilla qubits in the low-energy subspace of the simulating 2-local Hamiltonian.

Our third result is the \BQP-hardness for physically motivated 2-local Hamiltonians
(Proposition~\ref{prop:non2sld} and Proposition~\ref{prop:AF}).
The main obstacle here is that ground states of physically motivated 2-local Hamiltonians are not known to be guided by semi-classical subset states.
To solve the problem, we introduce another class of semi-classical states which we call {\it semi-classical encoded states} (see Definition~\ref{def:semiclassical} in Section~\ref{sec:preliminaries}).
Intuitively, semi-classical encoded states are states constructed from semi-classical subset states by applying a local isometry on each qubit.
Although semi-classical encoded states are more general than semi-classical subset states, they still allow succinct descriptions and efficient classical sampling algorithms (Lemma~\ref{lemma:sampling}).
For us, it is essential that semi-classical encoded states are closed under the applications of the local encoding of states during the perturbative simulations.
We show that semi-classical encoded states indeed satisfy this property, and therefore can guide ground states of physically motivated 2-local Hamiltonians.
The semi-classical encoded states newly introduced in this paper are of independent interest, and seem to have many other interesting applications.

Finally, our fourth result is to extend the $k$-GLH problem to the question of excited state energy estimation, we call this the \emph{Guided k-Local Hamiltonian Low Energy} ($k$-GLHLE) problem. In Ref.~\cite{PhysRevA.81.032331}, the authors show that determining the $c$th excited state energy of a $k$-local Hamiltonian ($k\geq 3$), where $c = \poly(n)$, is $\QMA$-complete -- even if all the $c-1$ energy eigenstates and corresponding energies are known. In their construction, they embed a $k$-local Hamiltonian $H$, encoding the QMA computation, in a Hamiltonian $H'$ living on a larger Hilbert space. This allows them to add up to polynomial number of  artificial eigenstates to $H'$ below the groundstate of $H$. Finding the $c$'th eigenvalue of $H'$ is then just as hard as finding the groundstate of $H$. We show that this construction translates to the setting with guiding states. 
%This follows from combining the results of the two previous sections with established results in Hamiltonian complexity theory, in particular the aforementioned $3$-local $\QMA$-complete excited state Hamiltonian from Ref.~\cite{PhysRevA.81.032331} and perturbative gadget techniques from Ref.~\cite{kempe2006}.
As a bonus, we also show that the unguided problem is $\QMA$-hard for $k=2$, which was left open in~\cite{PhysRevA.81.032331}.

\paragraph*{Open questions.} There are many open questions surrounding GLH, as well as the more general important goal of solving quantum chemistry problems on quantum computers. For example, we have shown \BQP-hardness of GLH for physically motivated Hamiltonians such as those with Heisenberg
%antiferromagnetic
interactions. An important next step would be to show \BQP-hardness for the specific types of fermionic Hamiltonians which are currently being studied in the quantum chemistry literature. Another subtle but important point is that, technically, the level of precision required for GLH in quantum chemistry scales as $1/n$, while
the hardness promise gap scales as $o(1/n)$
in \cite{gharibian2021dequantizing} and the present paper.
Can this be improved to $\Theta(1/n)$? A positive resolution to the quantum PCP conjecture would presumably, in turn, allows one to obtain hardness for gap $\Theta(1)$. Absent this, we are unaware of any circuit-to-Hamiltonian construction which is able to achieve $O(1/n)$ promise gap. Moreover, as mentioned earlier, the main bottleneck for quantum chemistry on quantum computers is the arduous task of \emph{finding} a good guiding state (if it even exists!). Can good heuristics be designed for this? Efforts to date suggest the answer so far is negative~\cite{C22}. Finally, more interestingly (but more challengingly), can one show \emph{rigorous} poly-time guiding-state computation algorithms for the specific families of Hamiltonians considered in the quantum chemistry literature?

\section{Preliminaries}\label{sec:preliminaries}
\paragraph*{Notation} We denote by $[M]$ the set $\{1,\dots,M\}$. We write $\lambda_i(A)$ to denote the $i$th eigenvalue of a Hermitian matrix $A$, ordered in non-decreasing order, with $\lambda_0(A)$ denoting the smallest eigenvalue (ground energy).  We denote $\text{eig}(A) = \{ \lambda_0(A),\dots,\lambda_{\text{dim}(A)-1}(A)\}$ for the (ordered) set of all eigenvalues of $A$.

\subsection{Semi-classical states}

In this section, we formally introduce the guided local Hamiltonian problem.
We first define two classes of semi-classical states.
The term ``semi-classical'' is motivated by the requirement for such states that they should be efficiently described (as an input of the problem) and efficiently samplable.\footnote{
The requirement of sampling access for a guiding state is motivated by the existence of an efficient classical algorithm for the GLH problem with constant precision, given a guiding state with sampling access, as shown in~\cite{gharibian2021dequantizing}.
%Efficient description of the guiding state is required as it should be given as input of the problem.
One type of a semi-classical state we use in this paper is a polynomial-size variant of the notion of subset states, first introduced in~\cite{grilo2015qma}.
}

\begin{definition}[Semi-classical subset state]\label{def:subsetstate}
We say that a normalized state $\ket{u} \in \mathbb{C}^{2^n}$ is a semi-classical subset state if there is a subset $S\subseteq \{0,1\}^n$ with $|S|=\poly(n)$ such that
\[
\ket{u}= \frac{1}{\sqrt{|S|}}\sum_{x\in S} \ket{x}.
\]
\end{definition}

A semi-classical subset state can be efficiently described by the description of $S$.
It is clear that we can efficiently sample from the probability distribution that outputs $x\in \{0,1\}^n$ with probability $|\langle x | u\rangle|^2$, i.e. according to the uniform distribution over $S$.

We next introduce a generalized version of a semi-classical subset state.

\begin{definition}[Semi-classical encoded state]\label{def:semiclassical}
We say that a normalized state $\ket{u} \in \mathbb{C}^{2^m}$, for $n<m \in\mathcal{O}(n)$, is a semi-classical encoded state if there is a subset $S\subseteq \{0,1\}^n$ with $|S|=\poly(n)$ and a set of isometries $V_1, V_2,...,V_n$, where each of $V_i$ maps a $1$-qubit state to an $\mathcal{O}(1)$-qubit state, such that
\[
\ket{u}= \frac{1}{\sqrt{|S|}}\sum_{x\in S} V_1(\ket{x_1})\otimes V_2(\ket{x_2}) \otimes \cdots \otimes V_n(\ket{x_n}).
\]
\end{definition}

A semi-classical encoded state is indeed a semi-classical subset state if the encoding is trivial (i.e. $V_1=V_2=\cdots=V_n=I$).
A semi-classical encoded state can be described by the description of $S$ and isometries $V_1,V_2...,V_n$.
We can also efficiently sample from the semi-classical encoded state as we show in the following lemma.

\begin{lemma}\label{lemma:sampling}
    Given the description of an $m$-qubit semi-classical encoded state $\ket{u}$, we can classically efficiently sample from the probability distribution
   that outputs $x\in\{0,1\}^m$ with probability $|\langle x|u\rangle|^2$.
\end{lemma}

\begin{proof}
Assume we are given the description, $S\subseteq\{0,1\}^n$ and $V_1,V_2,...,V_n$, of the semi-classical encoded state
\[
\ket{u}= \frac{1}{\sqrt{|S|}}\sum_{x\in S} V_1(\ket{x_1})\otimes V_2(\ket{x_2}) \otimes \cdots \otimes V_n(\ket{x_n}).
\]
Let $P(y_0,y_1,...,y_{i-1})=|(\bra{y_0,y_1,...,y_{i-1}}\otimes I)\ket{u}|^2$ be the probability that the  measurement outcome of the first $i$ qubits of $\ket{u}$ in the computational basis is $y_0,y_1,...,y_{i-1}$.
For each $i\in[m]$, we can efficiently calculate $P(y_0,y_1,...,y_{i-1})$ because $|S|=\poly(n)$ and $V_1(\ket{x_1})\otimes V_2(\ket{x_2}) \otimes \cdots \otimes V_n(\ket{x_n})$ is a product state of $\mathcal{O}(1)$-qubit states.
Then, we can also efficiently calculate the conditional probability
\[
P(z|y_0,y_1,...,y_{i-1})= \frac{P(y_0,y_1,...,y_{i-1},z)}{P(y_0,y_1,...,y_{i-1})}.
\]
If the bits $y_0,y_1,...,y_{i-1}$ have already been sampled, we compute $P(z|y_0,y_1,...,y_{i-1})$ and sample the next bit by tossing the coin with bias $P(0|y_0,y_1,...,y_{i-1})$.
In this way, we can classically efficiently sample from the probability distribution that outputs $x$ with probability $|\langle x|u\rangle|^2$.

\end{proof}

\subsection{Non-2SLD Hamiltonian and geometry of interaction}
To state the result,
we introduce some families of Hamiltonians.
Given a set of (at most) two-body interactions $\mathcal{S}=\{h_\alpha\}$, $\mathcal{S}$-Hamiltonian refers to the family of Hamiltonians that can be written in the form
\begin{equation}
\label{eq:s-hamiltonian}
H = \sum_{\langle i,j\rangle\in E} J_{i,j} h_{\alpha_{i,j}}^{(i,j)},
\end{equation}
where $J_{i,j} \in \mathbb{R}$, $h_{\alpha_{i,j}}^{(i,j)}$ is two-local interaction chosen from $\mathcal{S}$ and $E$ is the set of edges that represents the connectivity of interaction~\cite{cubitt2016complexity}.
If the connectivity of two-body interaction is restricted to a 2D square lattice, we call such a family $\mathcal{S}$-Hamiltonian on a 2D square lattice.
We also introduce the notion of 2SLD
%\footnote{The term 2SLD is short for ``the 2-local parts of all interactions
%in the set are simultaneously locally diagonalizable''.}
and non-2SLD:
\begin{definition}[2SLD interaction~\cite{cubitt2016complexity}]
\label{def:non2sld}
Suppose $\mathcal{S}$ is a set of interactions at most 2 qubits. We say that $\mathcal{S}$ is 2SLD if there exists $U \in {\rm SU}(2)$, such that for all $h_i \in \mathcal{S}$,
\[
U^{\otimes 2} h_i (U^\dagger)^{\otimes 2} = \alpha_i Z\otimes Z + A_i \otimes I + I \otimes B_i,
\]
where $\alpha_i \in \mathbb{R}$ and $A_i, B_i$ are arbitrary single-qubit Hamiltonians.
\end{definition}
A set $\mathcal{S}$ is non-2SLD if it is not 2SLD. In particular, such non-2SLD $\mathcal{S}$ includes the
following physically motivated\footnote{For clarity, in \cite{cubitt2016complexity} and here, all hardness results require \emph{non-uniform} weights on constraints. It is an open question whether one can obtain (say) QMA-hardness results with uniform (i.e. unit weight) constraints for such models. This remains an interesting open question, as many-body physicists typically utilize unit weights to model physical systems.} Hamiltonians:
\begin{itemize}
    \item $\{Z,X,Z Z,X X\}$
    \; ($ZZXX$ interaction~\cite{biamonte2008realizable})
   \item $\{Z,X,Z X,X Z\}$
    \; ($ZX$ interaction~\cite{biamonte2008realizable})
    \item $\{XX + Y Y\}$
    \; (general $XY$ interaction)
    \item $\{XX + YY + ZZ\}$ (general Heisenberg interaction).
\end{itemize}
If there is only a single type of interaction (like $\mathcal{S}=\{X X + Y  Y + ZZ\}$), the Hamiltonian is called {\it semi-translationally-invariant}. (Interaction strength can differ in each term.)

\paragraph*{Restriction on the sign of the interaction.}
We also introduce a further restricted class of $\mathcal{S}$-Hamiltonian in which all the signs of the coefficients are promised to be non-negative (i.e. all of $J_{i,j}$ in eq.~\eqref{eq:s-hamiltonian} must satisfy $J_{i,j}\geq 0$).
We call such a family of Hamiltonians as $\mathcal{S}^+$-Hamiltonian following~\cite{piddock2017complexity}.
In~\cite{piddock2017complexity}, the following results are shown:
\begin{itemize}
    \item $\{\alpha X X + \beta Y Y + \gamma Z Z\}^+$-Hamiltonian is QMA-complete if $\alpha+\beta>0$, $\alpha+\gamma>0$ and $\beta+\gamma>0$ hold.
    \item $\{\alpha XX + \beta Y Y + \gamma Z Z\}^+$-Hamiltonian is QMA-complete if the interactions are restricted to the edges of a 2D triangular lattice if
    $\alpha XX + \beta Y Y + \gamma Z Z$ is not proportional to $ X X + YY + Z Z$ in addition to the condition that
    $\alpha+\beta>0$, $\alpha+\gamma>0$ and $\beta+\gamma>0$ hold.
\end{itemize}
The first type of $\mathcal{S}^+$-Hamiltonian includes the antiferromagnetic Heisenberg model ($\{X X+Y Y + Z Z\}^+$-Hamiltonian) and the antiferromagnetic $XY$ model ($\{X X+Y Y \}^+$-Hamiltonian) as important special cases.
The antiferromagnetic $XY$ model (unlike the antiferromagnetic Heisenberg model) remains QMA-complete if
its geometric interaction is restricted to a 2D triangular lattice as it is included in the second type of $\mathcal{S}^+$-Hamiltonian above.

\section{GLHLE hardness constructions}\label{sec:complexity}

% In this section, we first show the \BQP-hardness of the 2-local GLH problem with $1-1/\poly(n)$ fidelity.

% \begin{theorem}
% %[Main, formal statement]
% \label{them:main}
% %\begin{reptheorem}{them:main}[Main result, formal statement]
% For any $\delta\in (0,1-\Omega(1/\poly(n)))$, there exist $a,b\in [0,1]$ with $b-a\in \Omega(1/\poly(n))$ and $0\leq c \leq \mathcal{O}(poly(n))$ such that $\text{GLHLE}(2,c,a,b,\delta)$  
% with a guiding semi-classical subset state is $\BQP$-hard.
% \end{theorem}

We next define the guided local Hamiltonian low energy (GLHLE) problem, which can be viewed as a generalization of GLH by considering arbitrary eigenstates of Hamiltonians\footnote{This definition of GLH is very similar to the definition of $\textup{GLH}^*(k,a,b,\delta)$ in~\cite{gharibian2021dequantizing}. The difference is that while the guiding states used in~\cite{gharibian2021dequantizing} are restricted to semi-classical subset states (Definition~\ref{def:subsetstate}), in our definition we use the more general concept of semi-classical
%\sout{samplable}
encoded states (Definition~\ref{def:semiclassical}).
%The efficient classical algorithm of~\cite{gharibian2021dequantizing} for $k$-GLH with $\beta-\alpha\in\Theta(1)$ and $\delta\in \Theta(1)$ only uses sampling access to the guiding state.
%Therefore, we think the usage of a broader definition of the semi-classical state does not significantly change the difficulty of the GLH problem.
Note that our \BQP-hardness result for general 2-local Hamiltonians (Proposition~\ref{prop:locality}) actually holds even when the guiding state is a semi-classical subset state. Proposition~\ref{prop:locality}, which optimally improves both the locality and fidelity parameters of \cite{gharibian2021dequantizing}, therefore holds in exactly the same setting as~\cite{gharibian2021dequantizing}. We use semi-classical encoded states only to show \BQP-hardness for further restricted families of Hamiltonians (Propositions~\ref{prop:non2sld} and~\ref{prop:AF}).}. For an $n$-qubit Hamiltonian $H$, we denote $\Pi_c$ the projector onto the space spanned by the states of $H$ that have energy $\lambda_c(H)$.

\begin{definition} [Guided Local Hamiltonian Low Energy] $\glh$
\label{def:GLHLE}
\begin{description}
\item[Input:] A $k$-local Hamiltonian $H$ on $n$ qubits such that $\|H\|\leq 1$ and the description of a semi-classical encoded state $\ket{u}\in \mathbb{C}^{2^n}$, a constant $c\in \mathbb{N}_\geq{0}$.
\item[Promise:]  $\|\Pi_{c}\ket{u}\|^2 \geq \delta$,
where $\Pi_{c}$ denotes the projection on the subspace spanned by the $c$th eigenstates, ordered in order of non-decreasing energy, of $H$, and either $\lambda_c(H) \leq a$ or $\lambda_c(H) \geq b$ holds.
\item[Goal:] Decide whether $\lambda_c(H)\leq a$ or $\lambda_c(H) \geq b$.
\end{description}
\end{definition}
The proof of Theorem~\ref{thm:main} consists of five parts: first, we show that 5-local GLH with $\delta = 1-\Omega(1/\poly(n))$ fidelity is $\BQP$-hard. 
Then, we show how to extend this result to the $\BQP$-hardness of the 6-local $\text{GLHLE}$ problem. 
Next, we improve the locality parameter and show a reduction from 6-local $\text{GLHLE}$ to 2-local $\text{GLHLE}$. Simultaneously we show that this also holds when we restict the Hamiltonians to be non$-2SLD$ $\mathcal S$-Hamiltonian on a $2D$ square lattice.  Finally, we show that $\BQP$-hardness persists if we restrict the family of Hamiltonians to be $\{XX + YY + ZZ\}^+$-Hamiltonians, or $\{XX + YY\}^+$-Hamiltonians on a $2D$ triangular lattice.

We state these five parts as propositions and prove them one by one, from this our main result (Theorem~\ref{thm:main}) follows. 

\subsection{Increasing the allowed fidelity}
The first proposition focuses on increasing the allowed fidelity of the guiding state with the ground state of the Hamiltonian of interest (and hence $c=0$). 
\begin{proposition}
\label{proposition:overlap}
For any $\delta\in (0,1-\Omega(1/\poly(n)))$, there exist $a,b\in [0,1]$ with $b-a\in \Omega(1/\poly(n))$ such that the problem $\text{GLHLE}(5,0,a,b,\delta)$ is $\BQP$-hard.
Moreover, it is still $\BQP$-hard with the additional two promises that
\begin{enumerate}
    \item $H$ has a non-degenerate ground state separated from the first excited state by a spectral gap $\gamma\in \Omega(1/\poly(n))$ in both the cases $\lambda_0(H) \leq a$ and $\lambda_0(H) \geq b$.(We call such instances $\gamma$-gapped $\GLHdes(k,a,b,\delta)$.)
    \item  The guiding state is restricted to be a semi-classical subset state.
\end{enumerate}
\end{proposition}
\begin{proof}
Let $\Pi=(\Pi_{\rm yes},\Pi_{\rm no})$ be a promise problem in $\BQP$, and $x\in \{0,1\}^n$ be an input.
Let $U_x= U_m U_{m-1}...U_1$ be a quantum circuit that decides $x$ consisting of $m=\poly(n)$ gates.
$U_x$ acts on $\ket{x}_A\otimes \ket{0...0}_B$ where $A$ denotes the $n$-qubit input register and $B$ denotes the poly-size ancilla register.
By measuring the output register of $U_x\ket{x}_A\otimes \ket{0...0}_B$, the quantum verifier outputs $1$ with probability at least $\alpha$ if $x\in \Pi_{\rm YES}$ (at most $\beta$ if $x\in \Pi_{\rm NO}$, respectively).
We may assume $\alpha = 1-2^{-n}$ and $\beta= 2^{-n}$ via the standard error reduction for $\BQP$.

Consider a pre-idled quantum verifier $\Tilde{U}_x\coloneqq U_xI\cdots I$, where $I$ is the identity gate.
The $\Tilde{U}_x$ consists of $M\coloneqq m + N$ gates, where $N$ is the number of idling steps. ($N=\poly(n)$ is taken properly later.)
Consider Kitaev's~\cite{Kitaev+02} 5-local circuit-to-Hamiltonian construction with an additional scaling factor:
\begin{align}
\label{eq:clock_construction}
H\coloneqq \Delta (\hin+\hprop+\hstab)+ \hout.
\end{align}
Here,
\begin{align}
    &\hin\coloneqq (I-\ket{x}\bra{x})_A\otimes (I-\ket{0...0}\bra{0...0})_B\otimes (\ket{0}\bra{0})_C\\
    &\hout\coloneqq \ket{0}\bra{0}_{\textup{out}} \otimes \ket{M}\bra{M}_C \\
    &\hstab\coloneqq \sum_{j=1}^{M-1} \ket{0}\bra{0}_{C_j} \otimes \ket{0}\bra{0}_{C_{j+1}}\\
    &\hprop\coloneqq \sum_{t=1}^M {H_t}, \rm{where} \\
    &\; H_t \coloneqq -\frac{1}{2}U_t\otimes \ket{t}\bra{t-1}_C -\frac{1}{2} U_t^\dagger\otimes \ket{t-1}\bra{t}_C
    + \frac{1}{2}I\otimes (\ket{t}\bra{t}_C + \ket{t-1}\bra{t-1}_C).
\end{align}

It is known that the non-degenerate and zero-energy ground space of $H_0\coloneqq \hin+\hprop+\hstab$ is spanned by $\ket{\psihist}$, where
\[
\ket{\psihist} \coloneqq
\frac{1}{\sqrt{M+1}} \sum_{t=0}^M \tilde{U_t}\tilde{U}_{t-1}\cdots \tilde{U}_1 \ket{x}_A\otimes\ket{0...0}_B\otimes \ket{t}_C.
\]
 It is also known that the smallest non-zero eigenvalue of $H_0$ is larger than $\pi^2/(64M^2)$ {\cite[Lemma 2.2]{Gharibian+19}} (based on \cite[Lemma 3]{GK12}).

We apply the Schrieffer-Wolf transformation for this $H$ by taking sufficiently large~$\Delta$.
%For the Schrieffer-Wolf transformation, see Appendix~\ref{section:SWtransform}.
Note that $\hout=\ket{0}\bra{0}\otimes I \otimes \ket{M}\bra{M}$ and $\|\hout\|=1$.
We would take
\[
\Delta \geq 16\cdot 64M^2/\pi^2.
\]
Then, $H$ has a one-dimensional ground space spanned by a ground state $\ket{g}$.
In the following, we analyze the fidelity between $\ket{g}$ and $\ket{\psihist}$, and the eigenvalue of $\ket{g}$ in the YES and NO cases.

\paragraph*{Analysis of fidelity.}
Using Equation
~\eqref{eq:SWoverlap} 
of Appendix
~\ref{section:SWtransform}
, the bound
\[
\|\ket{g}-\ket{\psihist}\| \in \mathcal{O}\left( (\Delta/M^2)^{-1} \right)
\]
holds.
Let us introduce the following state:
\[
\ket{u}\coloneqq \frac{1}{\sqrt{N}} \sum_{t=1}^{N}\ket{x}_A\otimes\ket{0...0}_B\otimes \ket{t}_C.
\]
This is a semi-classical subset state.
This state satisfies
\[
|\langle u|\psihist\rangle|^2 = \frac{N}{m+N+1}.
\]
Therefore, for any positive polynomial $r$, we can take sufficiently large $N, \Delta \in \mathcal{O}(\poly(n))$ so that
$|\langle u|g\rangle|^2 \geq 1-1/r(n)$.

\paragraph*{Analysis of eigenvalue.}
Next, we see the ground state energy of $H$ in both the YES case and the NO case.
The first-order effective Hamiltonian is given by
\[
H_{\rm{eff},1} = \ket{\psihist}\bra{\psihist} \hout \ket{\psihist}\bra{\psihist}.
\]
The history state is defined as
\[
\ket{\psihist}=\frac{1}{\sqrt{M+1}}
\sum_{t=1}^{M}\tilde{U}_t\cdots \tilde{U}_1\ket{x}_A\otimes \ket{0...0}_B\otimes \ket{t}_C
\]
and
\[
\bra{\psihist}\hout\ket{\psihist}
=
\frac{1}{M+1}\bra{x,0}U_x^\dagger (\ket{0}\bra{0}_{\textup{out}}\otimes I)U_x\ket{x,0}.
\]
The eigenvalue of $H_{\rm{eff},1}$ is given by  $\bra{\psihist}\hout\ket{\psihist}$ and
this is $\mathcal{O}((\Delta/M^2)^{-1})=\mathcal{O}(1/\poly(n))$-close to the ground state energy of $H$
using Equation 
\eqref{eq:eigenvalue}.

It can be verified that
$\bra{\psihist}\hout\ket{\psihist}\leq (1-\alpha)/(M+1)$ if $U_x$ accepts $x$ with probability at least $\alpha$ and
$\bra{\psihist}\hout\ket{\psihist}\geq (1-\beta)/(M+1)$ if $U_x$ accepts $x$ with probability at most $\beta$.
As we have mentioned earlier, we can assume $\alpha=1-2^{-n}$ and $\beta=2^{-n}$.
Therefore, the ground state energy $a$ of $H$ lies in the range of $0\pm \mathcal{O}((\Delta/M^2)^{-1})$ if $x\in \Pi_{\textup{yes}}$ and the ground state energy $b$ of $H$ lies in the range of $1/(M+1)\pm \mathcal{O}((\Delta/M^2)^{-1})$ if $x\in \Pi_{\textup{no}}$.

We also see the spectral gap between the ground state and any excited state in both the YES and NO cases.
We first see the NO case.
As we have shown, the ground state energy lies in $\frac{1-2^n}{M+1}\pm \mathcal{O}((M^2/\Delta))$.
In
$H= \Delta (\hin+\hprop+\hstab)+ \hout$,
the eigenvalues of $\Delta (\hin+\hprop+\hstab)$ is perturbed at most $\|\hout\|=1$. Therefore, the smallest non-zero eigenvalue of $H$ is larger than $(\Delta \pi^2)/(64M^2)-1$.
The spectral gap in the NO case is therefore
\[
\mathcal{O}\left(\frac{\Delta}{M^2}\right)-1-\left(\frac{1-2^{-n}}{M+1}+\mathcal{O}\left(\frac{M^2}{\Delta}\right)\right).
\]
The ground state energy in the YES case is smaller than that in the NO case.
Therefore, we can take sufficiently large $\Delta\in\poly(n)$ so that $H$ has inverse-polynomial spectral gap and $b-a\in \Omega(1/\poly(n))$.
Finally, we can normalize $H$ by a polynomially large factor, which concludes the proof. %\fbox
\end{proof}
\subsection{Extending to excited states}
The next proposition extends the result to excited states, at the cost of increasing the locality of the construction by one.
\begin{proposition}
\label{proposition:Low Enery states}
 For any $\delta \in \Omega(0,1-1/\text{poly}(n))$ there exist $a,b \in [-1,1]$ with $b-a \in \Omega(1/\text{poly}(n))$ and some number $0 \leq c \leq \text{poly}(n)$ such that $\text{GLHLE}(6,c,a,b,\delta)$ is \BQP-hard even when,
 \begin{enumerate}
     \item the $c$'th eigenvalue of $H$, $\lambda_c(H)$, is non-degenerate and is separated by a gap $\gamma \in \Omega(1/\text{poly}(n))$ from both $\lambda_{c-1}(H)$ and $\lambda_{c+1}(H)$. (We call such instances $\gamma$-gapped \\$\text{GLHLE}(k,c,a,b,\delta)$.)
     \item The guiding state is restricted to be a semi-classical subset state.
 \end{enumerate}
\end{proposition}

\begin{proof}
We will reduce directly from the $\BQP$-complete Hamiltonian $H$ as defined in Eq.~\eqref{eq:clock_construction}. Again, let $\ket{u}$ be a semi-classical guiding state such that $\abs{\braket{u}{\psi_0}} \geq \zeta$. Consider the following $6$-local Hamiltonian $H^{(c)}$ on $n+1$ qubits\footnote{Note that this gadget can be trivially changed such that estimating the $n$ highest energy states is $\BQP$-hard.}:
\begin{align}
    H^{(c)} = H^{(z)} \otimes \ket{0} \bra{0} +  H^{(s)} \otimes \ket{1}\bra{1},
    \label{eq:Hc}
\end{align}
where 
\begin{align*}
    &H^{(z)} = \sum_{i=0}^{d} 2^{i} \kb{1}{1}_i + \sum_{i=d+1}^n 2^{d+1}  \kb{1}{1}_i  -\left( c-\frac{1}{2}\right)I,\\
    &H^{(s)} =  \frac{1}{2} \frac{H+I/4}{\norm{H}+1/4} -\frac{1}{4} I,
\end{align*}
where we have that $d= \lceil \log_2(c) \rceil$. $H^{(z)}$ has exactly $c$ states with negative energy, with the smallest eigenvalue being $-c+\frac{1}{2}$ and the largest eigenvalue value at $\sum_{i=0}^d 2^i + \sum_{i=d+1}^{n} 2^{d+1} - \left(c-\frac{1}{2} \right) =  2^{d+1}  + 2^{d+1} (n - d) -\frac{1}{2}-c$. 
The spectrum jumps in integer steps of $1$, and has as largest negative (resp. smallest non-negative) energy value $-\frac{1}{2}$ (resp. $\frac{1}{2}$). Since $\text{eig}(H^{(s)})\in[-1/4,1/4]$, we must have that $H^{(s)}$ sits precisely at the $c$'th excited state level (or $c+1$'th eigenstate level) in $H^{(c)}$. Therefore, given a guiding state $\ket{u}$ for $H$ such that $\abs{\bra{u}{\psi_0}\rangle} \geq \delta$, one has that the guiding state $\lvert u^{(c)} \rangle  = \ket{u} \otimes  \ket{1}$ is also semi-classical and must have $| \langle u^{(c)} \lvert \psi_c^{(c)} \rangle| \geq \delta$, where $\lvert \psi_c^{(c)} \rangle$ denotes the $c$th excited state of $H^{(c)}$. Since this construction of $H^{(c)}$ and $\lvert u^{(c)} \rangle $ provides a polynomial time reduction from  an instance of $\text{GLH}(k,a,b,\delta)$  to one of $\text{GLHLE}(k,c,a,b,\delta)$, whenever $c=\mO(\poly(n))$, we must have that $\text{GLHLE}(k,c,a,b,\delta)$  is $\BQP$-hard whenever $k\geq 6$. The gap between $\lambda_{c}(H^{(c)}) - \lambda_{c-1}(H^{(c)})  = \frac{1}{4}$ and the gap between    $\lambda_{c+1}(H^{(c)}) - \lambda_{c}(H^{(c)})  = \gamma$ as before. The norm of the new Hamiltonian is bounded by $|H^{(c)}| = \mO(\text{poly}(n))$, hence after normalisation we retain $\lambda_{c}(H^{(c)}) - \lambda_{c-1}(H^{(c)})  \geq \lambda_{c+1}(H^{(c)}) - \lambda_{c}(H^{(c)}) = \Omega(1/\text{poly}(n))$.
\qedhere
\end{proof}

\subsection{Locality reduction and reduction to physically motivated Hamiltonians via strong Hamiltonian simulation}
The next two propositions bring (i) the locality $k$ down to $2$ and (ii) extend the result to any of non-2SLD $\mathcal{S}$-Hamiltonian on a 2D square lattice.
\begin{proposition}
\label{prop:locality}
Any $\gamma$-gapped $\text{GLHLE}(k,c,a,b,\delta)$ with $k\in \mathcal{O}(1)$, $b-a\in \Omega(1/\poly(n))$, $\delta \in (0,1-\Omega(1/\poly(n)))$, $0 \leq c \leq \text{poly}(n)$, and $\gamma \in \Omega(1/\poly(n))$
with a guiding  semi-classical subset state
can be reduced to $\gamma'$-gapped $\text{GLHLE}(2,c,a',b',\delta')$ with $b'-a'\in \Omega(1/\poly(n))$, $\delta' \in (0,1-\Omega(1/\poly(n)))$ and $\gamma'\in \Omega(\poly(n))$, and with a guiding semi-classical subset state in polynomial time.
\end{proposition}

% Theorem~\ref{them:main} immediately follows from these {\color{red}three} propositions.
% %The reduction of Proposition~\ref{prop:locality} applies for arbitrary $\mathcal{O}(1)$-local instances of the gapped GLH problem.
% In Section \ref{subsec:proposition1} we give a proof of Proposition \ref{proposition1} based on the Schrieffer-Wolf transformation (we give another proof of this proposition, based on the projection lemma from \cite{KKR06}, in Appendix~\ref{appendix:anotherproof}). The proof of Proposition \ref{proposition2} is given in Section \ref{subsec:proposition2}.

% We also show the \BQP-hardness of $\text{GLHLE}$ for the non-2SLD $\mathcal{S}$-Hamiltonian.

% \begin{theorem}\label{them:non2sld}
% For any $\delta\in (0,1-\Omega(1/\poly(n)))$, there exist $a,b\in [0,1]$ with $b-a\in \Omega(1/\poly(n))$ and some number $0 \leq c \leq \text{poly}(n)$ such that the problem $\text{GLHLE}(2,c,a,b,\delta)$
% with $b-a\in\Omega(1/\poly(n)) $ is $\BQP$-hard
% for Hamiltonians that are restricted to any of non-2SLD $\mathcal{S}$-Hamiltonian on a 2D square lattice.
% \end{theorem}

% Theorem~\ref{them:non2sld} follows by combining Proposition~\ref{proposition1} and  Proposition~\ref{proposition3} below:

\begin{proposition}\label{prop:non2sld}
Any $\gamma$-gapped $\text{GLHLE}(k,c,a,b,\delta)$ with $k\in \mathcal{O}(1)$, $b-a\in \Omega(1/\poly(n))$, $\delta \in (0,1-\Omega(1/\poly(n)))$, $0 \leq c \leq \text{poly}(n)$ and $\gamma \in \Omega(1/\poly(n))$, and with a guiding  semi-classical subset state
can be reduced to $\gamma'$-gapped $\text{GLHLE}(2,c,a',b',\delta')$ with $b'-a'\in \Omega(1/\poly(n))$, $\delta' \in (0,1-\Omega(1/\poly(n)))$ and $\gamma'\in \Omega(\poly(n))$ in polynomial time
whose Hamiltonian is restricted to any of non-2SLD $\mathcal{S}$-Hamiltonian on a 2D square lattice.
\end{proposition}
\begin{proof}[Proof of Propositions~\ref{prop:locality} and~\ref{prop:non2sld}]
Let $H$ and $\ket{u}$ be arbitrary inputs of $\text{GLHLE}(k,c,a,b,\delta)$ with $k\in \mathcal{O}(1)$, $b-a\in \Omega(1/\poly(n))$, $\delta \in (0,1-\Omega(1/\poly(n)))$.
From Theorem
~\ref{them:stronguniversality} 
(in Appendix
~\ref{app:apxHsim}
), 
we can efficiently find a non-2SLD $\mathcal{S}$-Hamiltonian $H'$ on a 2D square lattice that is a strong $(\Delta,\eta,\epsilon)$-simulation of $H$ given the description of $H$.
We take $\epsilon < (b-a)/2$, $b'=b-\epsilon$, $a'=a+\epsilon$ and $\Delta= \mathcal{O}(\epsilon^{-1}\|H\|^2+\eta^{-1}\|H\|)$ so that $\lambda_c(H')\leq a'$ if $\lambda_c(
H)\leq a$, and $\lambda_c({H'})\geq b'$ if $\lambda_c({H})\geq b'$ while $b'-a' \in \Omega(1/\poly(n))$.

We have shown the {\it existence} of desirable eigenvectors in the simulated Hamiltonian.
What remains to show is that
(i) the encoded state of $\ket{u}$ still has $1-1/\poly(n)$ fidelity with $c$'th excited state of $H'$ and
(ii) the encoded state is still a semi-classical subset state after the simulation by a 2-local Hamiltonian (for concluding Proposition~\ref{prop:locality})
, and
(iii) the encoded state is still a semi-classical encoded state after the simulation by an arbitrary non-2SLD $\mathcal{S}$-Hamiltonian on a 2D square lattice (for concluding Proposition~\ref{prop:non2sld}).

%In the simulation, the following four types of encoding are used.
%\begin{itemize}
 %   \item Attaching polynomially many ancilla qubits (mediator qubits) such as $\ket{0}$, $\ket{1}$, $\ket{+_y}$ (+1 eigenvector of $Y$), and $(\ket{00}-\ket{11})/\sqrt{2}$.
 %   \item Attaching a poly-size subset state.
 %   \item Apply local unitaries that only act on $\mathcal{O}(1)$ qubits.
 %   \item Encode each qubit into an $\mathcal{O}(1)$ qubit state (Subspace encoding).
%\end{itemize}

\paragraph*{(i) Verification of the fidelity.}
The fidelity can be analyzed by the following lemma:
\begin{lemma}[Simulation of the gapped excited state]
\label{lemma:groundstates}
Suppose the $c$'th excited state $\ket{g}$ of $H$ is non-degenerate and separated from both the $c-1$'th excited state and $c+1$'th excited state by a gap $\gamma$. 
Suppose $H'$ is a $(\Delta, \eta, \epsilon)$-simulation of $H$ such that $2\epsilon<\gamma$. 
Then $H'$ has a non-degenerate $c$'th excited state $\ket{g'}$ and
\[
\| \mathcal{E}_{\textup{state}}(\ket{g}) - \ket{g'} \| \leq \eta + \mathcal{O}(\gamma^{-1}\epsilon).
\]
\end{lemma}
%The proof of Lemma~\ref{lemma:groundstates} 
%is given in Appendix~\ref{proof:lemma_fidelity}
%.
%\subsection{Proof of Lemma~\ref{lemma:groundstates}}
%\label{proof:lemma_fidelity}

\begin{proof}
  This is a slight modification of Lemma 2 of \cite{bravyi2017complexity}. 
    First, the non-degeneracy of the $c$'th excited state of $H'$ follows because the $i$'th smallest eigenvalues of $H$ and $H'$ differs at most $\epsilon$ for all $0\leq i\leq {\rm dim}(H)-1$, and $\epsilon $ satisfies $2\epsilon<\gamma$. 
    Consider $H$ as an unperturbed Hamiltonian and $V\coloneqq \Tilde{\mathcal{E}}^\dagger H'\Tilde{\mathcal{E}}-H$ as a perturbation. Then, the perturbed Hamiltonian $H+V= \Tilde{\mathcal{E}}^\dagger H'\Tilde{\mathcal{E}}$ has a non-degenerate $c$'th excited state $\Tilde{\mathcal{E}}_{\textup{state}}(\ket{g'})$.
    The first-order perturbation theory for eigenvectors gives $\|\ket{g}-\Tilde{\mathcal{E}}_{\textup{state}}^\dagger(\ket{g'})\|\in\mathcal{O}(\gamma^{-1}\epsilon)$. 
    Therefore, it follows that $\|\Tilde{\mathcal{E}}_{\textup{state}}(\ket{g})-\ket{g'}\|=\|\Tilde{\mathcal{E}}_{\textup{state}}(\ket{g})-\Tilde{\mathcal{E}}_{\textup{state}}(\Tilde{\mathcal{E}}_{\textup{state}}^\dagger(\ket{g'}))\|\in \mathcal{O}(\gamma^{-1}\epsilon)$
    using that $\tilde{\mathcal{E}}_{state}$ is an isometry and $\ket{g'}\in {\rm Im}(\tilde{\mathcal{E}}_{state})$
    . Finally, by using $\|\mathcal{E}_{state}-\tilde{\mathcal{E}}_{state}\|\leq \eta$,
    $
\| \mathcal{E}_{\textup{state}}(\ket{g}) - \ket{g'} \| \leq \eta + \mathcal{O}(\gamma^{-1}\epsilon)
$ follows.
\end{proof}

Using Lemma~\ref{lemma:groundstates},
we can take sufficiently small $\epsilon$ and $\eta$ to ensure $\| \mathcal{E}_{\textup{state}}(\ket{u}) - \ket{g'} \| \leq \delta' = \delta - 1/\poly(n)$.
Because the Hamiltonian simulation is efficient, the operator norm $\|H'\|$ and the number of qubits of $H'$ is in $\poly(n)$.

\paragraph*{(ii) Verification of the semi-classical property for Proposition~\ref{prop:locality}.}
We start from a semi-classical subset state  $\ket{u}=1/\sqrt{|S|}\sum_{x\in S}\ket{x}$.
We show that after the simulation of the original $k$-local Hamiltonian $H$ where $k\in\mathcal{O}(1)$ by an $2$-local Hamiltonian, the corresponding encoding $\mathcal{E}_{\textup{state}}(\ket{u})$ is still a semi-classical subset state.

In order to simulate the $k$-local Hamiltonian by a $2$-local Hamiltonian (that has no restriction on the family of Hamiltonian), it is enough to use mediator qubit gadgets that attach $\ket{0}$ states for mediator qubits (called subdivision and 3-to-2 gadgets~\cite{oliveira2005complexity}).
%to show Proposition~\ref{proposition2} that has no restriction on the family of Hamiltonian.
A $k$-local term can be simulated by $(\lceil k/2 \rceil +1)$-local terms using the subdivision gadget.
Moreover, subdivision gadgets can be applied to each of the terms of the Hamiltonian in parallel~\cite{piddock2017complexity, cubitt2018universal}.
Therefore, we can reduce a $k$-local Hamiltonian to a 3-local Hamiltonian by $\mathcal{O}(\log{k})$ rounds of applications of the subdivision gadgets.
Then we can use the 3-to-2 gadgets in parallel to reduce to a 2-local Hamiltonian.
In the corresponding encoding of states of this procedure, polynomially many $\ket{0}$ states are attached to the original state.
Clearly, by attaching polynomially many $\ket{0}$ states, a polynomial-size subset state is mapped to another polynomial-size subset state:
\[
\frac{1}{\sqrt{|S|}}\sum_{x\in S}\ket{x}\rightarrow \frac{1}{\sqrt{|S|}}\sum_{x\in S}\ket{x}\ket{0}^{\otimes \poly(n)}=
\frac{1}{\sqrt{|S|}}\sum_{x\in S\times \{0...0\}}\ket{x}
.
\]
This concludes the proof of Proposition~\ref{prop:locality}.

\paragraph*{(iii) Verification of the semi-classical property for Proposition~\ref{prop:non2sld}.}
We proceed to show that starting from a semi-classical subset state $\ket{u}$, the resulting state is a semi-classical encoded state when we simulate the original Hamiltonian by a non-2SLD $\mathcal{S}$-Hamiltonian on a 2D square lattice.
%{\color{red}
%The flow of perturbative simulation is given in Appendix
%~\ref{sec:encoding}
%.}
There are three types of encodings used in the simulation:
\begin{itemize}
    \item {\bf Mediator qubits.} In this encoding, some simple ancilla states are attached to the original state.
    \item {\bf Subspace encoding.} In this encoding, a local isometry is applied to the original state.
    \item {\bf Local Unitaries.} In this encoding, local unitary $U\otimes U \otimes \cdots \otimes U$, where each of $U$ acts on one qubit, is applied to the original state.
\end{itemize}
We restate the chain of Hamiltonian simulations of Appendix
~\ref{sec:encoding}
:
%\begin{spacing}{1.2}
\begin{screen}
{\bf Arbitrary $k$-local Hamiltonian}\\
\ $\downarrow$ \hspace{5mm}
{(1) Mediator qubits. }
(Attach a semi-classical subset state $\ket{\alpha}$.)
\\
{\bf Spatially sparse $5$-local Hamiltonian}\\
\ $\downarrow$ \hspace{5mm}
(2) Mediator qubits.
(Attach polynomially many $\ket{+_y}$ states.)
\\
{\bf Spatially sparse $10$-local real Hamiltonian}\\
\ $\downarrow$ \hspace{5mm}
(3) Mediator qubits.
(Attach polynomially many $\ket{0}$ or $\ket{1}$ states.)
\\
{\bf Spatially sparse $2$-local Pauli interactions with no $Y$-terms}\\
\ $\downarrow$ \hspace{5mm}
(4) Subspace encoding.
%Perform a subspace encoding.
\\
{\bf Spatially sparse  $\mathcal{S}_0=\{XX+YY+ZZ\}$ or $\{XX+YY\}$ Hamiltonian}\\
\ $\downarrow$  \hspace{5mm}
(5) Mediator qubits. (Attach polynomially many $\ket{0}$ or $\ket{1}$ states.)
%Attach polynomially many mediator qubit states.
\\
{\bf$\mathcal{S}_0$-Hamiltonians on a 2D square lattice}\\
\ $\downarrow$ \hspace{5mm}
(6) Mediator qubits, Subspace encoding, and local unitary.
%Perform mediator qubit
\\
{\bf Arbitrary non-2SLD $\mathcal{S}$-Hamiltonian on a 2D square lattice}
\end{screen}
%\end{spacing}
In step (1), a semi-classical subset state is attached to a semi-classical subset state $\ket{u}$.
The resulting state is also a semi-classical subset state:
\begin{align}
    \ket{u}=\frac{1}{\sqrt{|S|}}\sum_{x\in S}\ket{x} \rightarrow & \frac{1}{\sqrt{|S|}}\sum_{x\in S}\ket{x} \otimes \frac{1}{\sqrt{|S'|}}\sum_{x'\in S'}\ket{x'}\nonumber\\
    &=\frac{1}{\sqrt{|S||S'|}}\sum_{x\in S\times S'}\ket{x}.
\end{align}
The resulting state after the encodings of steps (2)$\sim$(4) is a semi-classical encoded state because in these steps, a tensor product of single-qubit states is attached to a semi-classical subset state and then the state is encoded by a local isometry.
By further performing a local encoding to the semi-classical encoded state, the resulting state is also a semi-classical encoded state.
This concludes the proof of Proposition~\ref{prop:non2sld}.
\end{proof}
Finally, we show a \BQP-hardness result for the antiferromagnetic Hamiltonian.
\begin{proposition}
\label{prop:AF}
For any $\delta\in (0,1-\Omega(1/\poly(n)))$, there exist $a,b\in [0,1]$ with $b-a\in \Omega(1/\poly(n))$ and $0\leq c \leq \mathcal{O}(poly(n))$ such that the problem $\text{GLHLE}(2,c,a,b,\delta)$ with $b-a\in\Omega(1/\poly(n)) $ is $\BQP$-hard
for Hamiltonians that are restricted to either
 $\{XX+YY+ZZ\}^+$-Hamiltonian, or
 $\{XX+YY\}^+$-Hamiltonian on a 2D triangular lattice.
\end{proposition}
\begin{proof}
We first prove the case of $\{XX+YY+ZZ\}^+$-Hamiltonian.
This can be reduced from the $\text{GLHLE}$ problem of $\{XX+YY+ZZ\}$-Hamiltonian with a semi-classical encoded state as a guiding state, which is shown to be \BQP-hard in Proposition~\ref{prop:non2sld}.
The
$\{XX+YY+ZZ\}$-Hamiltonian can be simulated by $\{XX+YY+ZZ\}^+$-Hamiltonian using the ``basic gadget'' (this is a type of a mediator qubit gadget) of~\cite{piddock2017complexity}.
In the corresponding encoding of the state, a tensor product of two-qubit states is attached to the original state. This encodes a semi-classical encoded state to another semi-classical encoded state. 
The reason is as follows. Let us denote 
%the original guiding state as 
%\[\ket{u}= \frac{1}{\sqrt{|S|}}\sum_{x\in S} V_1 \otimes V_2 \otimes \cdots \otimes V_n \ket{x_1}\otimes \ket{x_2}\cdots \otimes \ket{x_n}\]
the attached tensor product of polynomially many two-qubit states as 
\[
\ket{\phi_1}\otimes \ket{\phi_2} \otimes \cdots \ket{\phi_m} = V'_1\ket{0} \otimes V'_2\ket{0} \otimes \cdots V'_m \ket{0},
\]
where $\ket{\phi_1},...,\ket{\phi_m}$ are two-qubit states and $V'_1,...,V'_m$ are isometries such that $V'_i\ket{0}=\ket{\phi_i}$ for each $i\in[m]$. 
Then, the original semi-classical encoded state represented by a polynomial-size subset $S$ and a local isometry $V_1\otimes V_2 \otimes \cdots \otimes V_n$ is mapped to a semi-classical encoded state represented by a subset $S\times \{0...0\}$ and a local isometry $V_1\otimes \cdots \otimes V_n\otimes V'_1\otimes \cdots \otimes V'_m$. This concludes the case of $\{XX+YY+ZZ\}^+$-Hamiltonian.

We next show the \BQP-hardness of the $\text{GLHLE}$ problem of
 $\{XX+YY\}^+$-Hamiltonian on a 2D triangular lattice with a semi-classical encoded state.
 We show a reduction from the $\text{GLHLE}$ problem of $\{XX+YY\}$-Hamiltonian on a 2D square lattice with a semi-classical encoded state as a guiding state, which is shown to be \BQP-hard in Proposition~\ref{prop:non2sld}.
 It is shown in~\cite{piddock2017complexity} how to simulate  $\{XX+YY\}$-Hamiltonian on a 2D square lattice by  $\{XX+YY\}^+$-Hamiltonian on a 2D triangular lattice by using mediator qubit gadgets. The corresponding encoding is just attaching a product state of polynomially many $\mathcal{O}(1)$-qubit states to the original guiding state. Therefore, the original semi-classical encoded state is mapped to another semi-classical encoded state (by a similar reason as in the case of $\{XX+YY+ZZ\}^+$-Hamiltonian).
\end{proof}

\bibliography{lipics-v2021-sample-article}

\appendix

\section{Approximate Hamiltonian simulation}

\subsection{Introduction of approximate Hamiltonian simulation}
\label{app:apxHsim}
While in the QMA-hardness reduction it suffices to focus only on the eigenvalues in the simulation,
in the reduction of GLH it is also important to know how the eigenvectors change in the perturbative simulation.
It is convenient to introduce the notion of {\it approximate Hamiltonian simulation} to show the reduction of GLH.

\begin{definition}[Approximate Hamiltonian simulation~\cite{cubitt2018universal},~\cite{zhou2021strongly}]\label{def:approxsimulation}
We say that an $m$-qubit Hamiltonian $H'$ is a $(\Delta,\eta,\epsilon)$-simulation of an $n$-qubit Hamiltonian $H$ if there exists a local encoding $\mathcal{E}(M)=V(M\otimes P+ \bar{M}\otimes Q)V^\dagger$ such that
\begin{enumerate}
    \item There exists an encoding $\Tilde{\mathcal{E}}(M)=\tilde{V}(M\otimes P+ \bar{M}\otimes Q)\tilde{V}^\dagger$ such that
    $\tilde{\mathcal{E}}(\mathbbm{1})=P_{\leq \Delta(H')}$ and $\|\tilde{V}-V\|\leq \eta$, where $P_{\leq \Delta(H')}$ is the projector onto the subspace spanned by eigenvectors of $H'$ with eigenvalue below $\Delta$,
    \item $\|H'_{\leq \Delta}-\tilde{\mathcal{E}}(H)\|\leq \epsilon$, where $H'_{\leq \Delta}\coloneqq P_{\leq \Delta(H')}H'$.
\end{enumerate}
Here, $V$ is a local isometry that can be written as $V=\bigotimes_iV_i$ where each $V_i$ is an isometry acting on at most 1 qubit, and $P$ and $Q$ are locally orthogonal projectors (i.e. for all $i$ there exist orthogonal projectors $P_i$ and $Q_i$ acting on the same subsystem as $V_i$ such that $P_iQ_i=0$, $P_iP=P$ and $Q_iQ=Q$) such that $P+Q=I$, and $\bar{M}$ is the complex conjugate of $M$.
Moreover, we say that the simulation is efficient if $m$ and $\|H'\|$ are at most $\mathcal{O}(\poly(n,\eta^{-1},\epsilon^{-1},\Delta))$, and the description of $H'$ can be computable in $\poly(n)$ time given the description of $H$.
\end{definition}

We approximately simulate the original Hamiltonian $H$ in the low-energy subspace of $H'$.
There is a corresponding encoding of a state which can be taken as
\[
\mathcal{E}_{\textup{state}}(\rho)= V(\rho\otimes \sigma)V^\dagger
\]
for $\sigma$ such that $P\sigma=\sigma$ (if $P\neq 0$).
If $\rho$ is the eigenvector of $H$ with eigenvalue $\alpha$, then $\mathcal{E}_{\textup{state}}(\rho)$ is approximately the eigenvector of $H'$ with eigenvalue $\alpha'\in [\alpha-\epsilon,\alpha+\epsilon]$.

In \cite{zhou2021strongly}, it is shown that there exist families of Hamiltonians that can efficiently simulate any $\mathcal{O}(1)$-local Hamiltonians.
They call such families of Hamiltonians {\it strongly universal Hamiltonians}.\footnote{
It would be possible to show Theorem~\ref{thm:main} by modifying the verifier circuit $\tilde{U}_x$ following~\cite{oliveira2005complexity} to make the constructed Hamiltonian spatially sparse.
We believe Proposition~\ref{prop:locality} is interesting because the reduction holds for arbitrary $\mathcal{O}(1)$-local Hamiltonian even if it is not originally spatially sparse.
}
We use the construction of strongly universal Hamiltonians of \cite{zhou2021strongly} to show Proposition~\ref{prop:locality}.
Formally, the strong (and weak) universality is defined as follows:
\begin{definition}[Strong and weak universality \cite{zhou2021strongly}]
A family of Hamiltonians $\mathcal{H}=\{H_m\}$ is weakly universal if given any $\Delta,\eta,\epsilon >0$, any $\mathcal{O}(1)$-local, $n$-qubit Hamiltonian can be $(\Delta,\eta,\epsilon)$-simulated.
Such a family is strongly universal if the simulation is always efficient.
\end{definition}

The following result is shown in \cite{zhou2021strongly}:
\begin{theorem}[\cite{zhou2021strongly}]
\label{them:stronguniversality}
Any non-$2SLD$ $\mathcal{S}$-Hamiltonian on a $2D$-square lattice is strongly universal.
\end{theorem}

\section{Schrieffer-Wolf transformation for 1-dimensional gapped ground space}
\label{section:SWtransform}

Let us introduce the Schrieffer-Wolf transformation and its approximation \cite{bravyi2011schrieffer} which we use in the proof.
We only consider the case when the unperturbed Hamiltonian has 1-dimensional ground space.

Let $H_0$ be a Hamiltonian that has 1-dimensional ground space spanned by $\ket{g_0}$ whose energy is 0. Let us assume that the smallest non-zero eigenvalue of $H_0$ is larger than one.
Consider the following (perturbed) Hamiltonian:
$H=\Delta H_0 + V$.
We shall always assume that $\|V\|\leq \Delta/2$ in the following.
Then, there is only one eigenvector (which we denote $\ket{g}$) of $H$ with eigenvalue lying in the interval of $[-\Delta/2,\Delta/2]$ (Lemma 3.1 of \cite{bravyi2011schrieffer}).

Then, the Schrieffer-Wolf (SW) transformation is defined as
a unitary $U_{\rm SW}$ that maps the ground space of $H$ to that of $H_0$. That is, $U_{\rm SW}\ket{g}=\ket{g_0}$.
The Hamiltonian
\[
H_{\rm eff}=\Pi_0 U_{\rm SW} (\Delta H_0 + V)U_{\rm SW}^\dagger \Pi_0
\]
is called the effective low-energy Hamiltonian. Here, $\Pi_0$ is the projector onto the ground space of $H_0$.
The eigenvector of $H_{\rm eff}$ is $\ket{g_0}$ and the eigenvalue is the same as the eigenvalue of $\ket{g}$ with respect to $H$.

Next, we show how to approximate $U_{\rm SW}$ and $H_{\rm eff}$. We only need the simplest first-order approximation
in the proof of Proposition~\ref{proposition:overlap}.
In the following, we further assume $\|V\|\leq \Delta/16$. Then, it is known that
\begin{equation}
\label{eq:eigenvector}
\|I-U_{\rm SW}\|\in \mathcal{O}(\Delta^{-1}\|V\|)
\end{equation}
and
\begin{equation}
\label{eq:eigenvalue}
\|H_{\rm eff}-\Pi_0 V\Pi_0\|\in \mathcal{O}(\Delta^{-1}\|V\|^2)
\end{equation}
hold (Lemma 3.4  \cite{bravyi2011schrieffer}, Lemma 4  \cite{bravyi2017complexity}).
This means that $I$ and $\Pi_0V\Pi_0$ work as the first-order approximation of $U_{\rm SW}$ and $H_{\rm eff}$, respectively.
The derivation and the forms of the higher-order terms can be found in \cite{bravyi2011schrieffer}.
From eq.~\eqref{eq:eigenvector}, it follows that
\begin{equation}
\label{eq:SWoverlap}
\big\|\ket{g}-\ket{g_0}\big\|=
\left\|(I-U_{\rm SW}^\dagger)\ket{g_0}\right\|\in \mathcal{O}(\Delta^{-1}\|V\|).
\end{equation}
It follows from eq.~\eqref{eq:eigenvalue} that the ground state energy of $H$ differs at most $\mathcal{O}(\Delta^{-1}\|V\|^2)$ from the eigenvalue of $H_{\rm{eff},1}\coloneqq \Pi_0 V\Pi_0$ (restricted to the space spanned by $\ket{g_0}$).

\section{Encoding of states for strong Hamiltonian simulation}
\label{sec:encoding}

We sketch the construction of the strong Hamiltonian simulation introduced in \cite{zhou2021strongly}.
%and verify that the corresponding encoding of states preserves the semi-classical property.
The simulation mainly consists of two parts.
First, they construct spatially sparse 5-local Hamiltonian~\cite{oliveira2005complexity} using a quantum phase estimation circuit and its modification.
This procedure may be thought of as a ``Hamiltonian-to-circuit'' (then goes back to Hamiltonian by circuit-to-Hamiltonian) construction.
Then, they perturbatively simulate the spatially sparse Hamiltonian with known techniques in the literature~\cite{oliveira2005complexity, cubitt2018universal,piddock2017complexity}.
%It can be verified the semi-classical property is maintained throughout the encoding if the original state is a semi-classical state.
In the following, we overview their construction.

\paragraph*{(1) Arbitrary $k$-local Hamiltonian $\rightarrow$ spatially sparse $5$-local Hamiltonian (\cite{zhou2021strongly}).}
Let $H$ be a target $\mathcal{O}(1)$-local Hamiltonian. Assume that $H$ can be written as $H=\sum_i E_i \ket{\psi_i} \bra{\psi_i}$ where $\{E_i\}$ and $\{\ket{\psi_i}\}$ are the eigenvalues and eigenvectors of $H$.
In~\cite{zhou2021strongly}, they showed that there is a spatially sparse quantum circuit $U^{\rm sparse}_{\rm PE}$ that approximately estimates the energy of $H$, i.e.
\[
U^{\rm sparse}_{\rm PE} \sum_i c_i \ket{\psi_i}\ket{0^m} \approx \sum_i c_i \ket{\psi_i}\ket{\tilde{E}_i}\ket{other},
\]
where $\{c_i\}$ are arbitrary coefficients and $\{|\tilde{E}_i\rangle\}$ are approximations of $\{E_i\}$.

The circuit $U^{\rm sparse}_{\rm PE}$ is implemented first by constructing $U^{\rm sparse}_{\rm NN}$ that consists of 1D nearest-neighborhood interaction. Then, $U^{\rm sparse}_{\rm NN}$ is converted into a spatially sparse circuit using ancilla qubits and swap gates.

Then they combine uncomputation and idling to construct
\[U=({\rm Idling})(U^{\rm sparse}_{\rm PE})^\dagger({\rm Idling})U^{\rm sparse}_{\rm PE}.\]
They apply circuit-to-Hamiltonian construction for this $U$ to construct spatially sparse 5-local Hamiltonian $H_{\rm circuit}$.
They use first-order perturbation theory to show that $H_{\rm circuit}$ simulates $H$ in its low-energy subspace.
The encoding of $H_{\rm circuit}$ to the low energy subspace of $H$ is approximated by the map: $H \rightarrow H \otimes \ket{\alpha}\bra{\alpha}$.
Here,
$\ket{\alpha}$ is a subset state with $\poly(n)$-size subset $S'$ that is related to the history state of the idling steps after uncomputation. For detail, see the proof of Proposition 2 of~\cite{zhou2021strongly}.
Then, the corresponding encoding of the state is
\[
\ket{u} \rightarrow \ket{u}\otimes \ket{\alpha}.
\]
The encoded state is also a semi-classical subset state if $\ket{u}$ is a semi-classical subset state.

\paragraph*{(2) Spatially sparse $5$-local Hamiltonian $\rightarrow$ Spatially sparse $10$-local real Hamiltonian~(Lemma 22 of~\cite{cubitt2018universal}).}
In this simulation, the state is encoded by attaching polynomially many $\ket{+_y}$ where $\ket{+_y}$ is the $+1$ eigenvector of Pauli $Y$ matrix:
\begin{equation}
\label{eq:11}
\ket{u} \rightarrow \ket{u}\otimes\ket{+_y}\otimes \cdots \otimes\ket{+_y}.
\end{equation}
This encoding does not map a semi-classical subset state into a semi-classical state but maps into a semi-classical encoded state. The reason is as follows. Let $V_y$ be a unitary such that $\ket{+_y}=V_y\ket{0}$, and $\ket{u}=1/\sqrt{|S|}\sum_{x\in S}\ket{x}$.
Then, the right side of eq.~\eqref{eq:11} can be written as
\[
\ket{u}\otimes\ket{+_y}\otimes \cdots \otimes\ket{+_y} 
= \frac{1}{\sqrt{|S|}} \sum_{x\in S\times \{0...0\}} I\otimes  \cdots \otimes I \otimes V_y \otimes \cdots \otimes V_y \ket{x}.
\]
This is a semi-classical encoded state with a subset $S\times\{0...0\}$ and a local isometry (this is indeed a local unitary) $I\otimes  \cdots \otimes I \otimes V_y \otimes \cdots \otimes V_y$.

\paragraph*{(3) Spatially sparse $10$-local real Hamiltonian $\rightarrow$ Spatially sparse $2$-local Pauli interactions with no $Y$-terms~(\cite{oliveira2005complexity,cubitt2016complexity}).}
This can be done first by simulating the $10$-local real Hamiltonian with $11$-local Hamiltonian whose Pauli decomposition does not contain any Pauli $Y$ terms~\cite[Lemma 40]{cubitt2018universal}.
In the corresponding encoding, $\ket{1}$ states are attached for the polynomially many mediator qubits introduced in the simulation.
Then, we can use subdivision gadgets and 3-to-2 gadgets~\cite{oliveira2005complexity}.
In this simulation, polynomially many mediator qubits are introduced, and the encoding of states is just to add $\ket{0}$ states for each of the mediator qubits. The resulting Hamiltonian can be written in the form $\sum_{i<j}\alpha_{ij}A_{ij}+\sum_k (\beta_kX_k + \gamma_k Z_k)$, where $A_{ij}$ is one of the interactions of $X_iX_j$, $X_iZ_j$, $Z_iX_j$ or $Z_iZ_j$.

\paragraph*{(4) Subspace encoding for spatially sparse  $\mathcal{S}_0=\{XX+YY+ZZ\}$ or $\{XX+YY\}$ Hamiltonian (Theorem 42 of~\cite{cubitt2018universal}).}
We have already obtained 2-local Hamiltonian in the form $\sum_{i<j}\alpha_{ij}A_{ij}+\sum_k (\beta_kX_k + \gamma_k Z_k)$. Then we show how to simulate this Hamiltonian with arbitrary non-2SLD $\mathcal{S}$-Hamiltonians.
We first consider $\mathcal{S}_0$ Hamiltonian, where $\mathcal{S}_0=\{XX+YY+ZZ\}$ or $\mathcal{S}_0=\{XX+YY\}$.
In this simulation, we use subspace encoding in which the {\it logical qubit} of the original Hamiltonian is encoded into four {\it physical qubits}.
Consider the simulation by Heisenberg interaction $\{XX+YY+ZZ\}$ for example.
Each logical qubit is encoded into 4 qubit state by an isometry that is defined as
\begin{align}\label{eq:isometry1}
&V\ket{0}=\ket{0_L}=\ket{\Psi_-}_{13}\ket{\Psi_-}_{24}\\ \label{eq:isometry2}
&V\ket{1}=\ket{1_L}=\frac{2}{\sqrt{3}}\ket{\Psi_-}_{12}\ket{\Psi_-}_{34}-\frac{1}{\sqrt{3}}\ket{\Psi_-}_{13}\ket{\Psi_-}_{24},
\end{align}
where $\ket{\Psi_-}=(\ket{01}-\ket{10})/\sqrt{2}$.
For details, see \cite[Theorem 42]{cubitt2018universal}.
The encoding of states for {XX+YY} interaction is the same. 
A semi-classical encoded state is clearly mapped to a semi-classical encoded state by applying a local isometry of the corresponding subspace encoding.

\paragraph*{(5) Spatially sparse $\mathcal{S}_0$-Hamiltonian $\rightarrow$ $\mathcal{S}_0$-Hamiltonians on a 2D square lattice (Lemma 47 of~\cite{cubitt2018universal}).}
This simulation can be done using three perturbative gadgets called subdivision, fork, and crossing gadgets. All of these gadgets attach a mediator qubit for each use of the gadgets.
$\mathcal{O}(1)$ rounds of parallel use of perturbative gadgets are sufficient to simulate a spatially sparse $\mathcal{S}_0$-Hamiltonian by a $\mathcal{S}_0$-Hamiltonians on a 2D square lattice, which prevents the interaction strength to grow exponentially. (For general interaction graphs, $\mathcal{O}(\log{n})$ rounds of perturbative simulations are necessary.)

\paragraph*{(6) $\mathcal{S}_0$-Hamiltonian on 2D square lattice $\rightarrow$ Arbitrary non-SLD $\mathcal{S}$-Hamiltonian on a 2D square lattice (Theorem 43 of~\cite{cubitt2018universal}).}
Finally, this simulation is similarly done by using variants of mediator qubit gadgets or subspace encoding gadgets as well as applying
local unitaries.\footnote{Applying local unitaries means to simulate $H$ by $U^{\otimes n}H(U^\dagger)^{\otimes n}$ where $U$ acts on one qubit. The corresponding encoding of state is $\mathcal{E}_{state}(\ket{\psi})=U^{\otimes n}\ket{\psi}$.}

\enlargethispage{\baselineskip}

\end{document}